\documentclass[12pt]{article}
\usepackage[english]{babel}
\usepackage[letterpaper,top=2cm,bottom=2cm,left=3cm,right=3cm,marginparwidth=1.75cm]{geometry}
\usepackage{amsmath}
\usepackage{graphicx}
\usepackage[colorlinks=true, allcolors=blue]{hyperref}
\usepackage{rotating}
\usepackage{amssymb}
\usepackage{tikz}
\usepackage{tikz-network}
\usepackage{amsthm}
\usepackage{float}
\usepackage{nicefrac}
\usepackage{caption}
\usepackage{subcaption}
\theoremstyle{definition}
\newtheorem{theorem}{Theorem}
\usepackage{multirow}
\usepackage[toc,page]{appendix}
\usepackage{booktabs}
\usepackage{siunitx}  
\usepackage{float}    
\usepackage{makecell} 
\title{Deep Declarative Risk Budgeting Portfolios}
\author{Manuel Parra-Diaz\footnote{manuel.parra@urosario.edu.co}\; and Carlos Castro-Iragorri\footnote{carlos.castro@urosario.edu.co} } 
\begin{document}
\maketitle

\begin{abstract}
Recent advances in deep learning have spurred the development of end-to-end frameworks for portfolio optimization that utilize implicit layers. However, many such implementations are highly sensitive to neural network initialization, undermining performance consistency. This research introduces a robust end-to-end framework tailored for risk budgeting portfolios that effectively reduces sensitivity to initialization. Importantly, this enhanced stability does not compromise portfolio performance, as our framework consistently outperforms the risk parity benchmark.
\end{abstract}




\section{Introduction}
\label{sec:intro} 

Portfolio optimization is a crucial financial task in which investors aim to maximize returns while effectively managing risk. The End-to-End framework offers a promising approach to portfolio optimization by integrating deep learning techniques with traditional financial models \cite{Uysal2023}. This innovative architecture enables a comprehensive and efficient resolution of various steps in portfolio optimization, including parameter estimation and the determination of optimized weights. The decision-making process can be tailored to accommodate different investor objectives, making the approach highly flexible. Additionally, the framework is tractable, potentially scalable, and provides numerical sensibilities to all elements of the problem. Its features can be implemented using data-driven insights, model-based decisions, or a combination of both, significantly enhancing portfolio allocation efficiency for both institutional and individual investors.


Despite its advantages, implementing the End-to-End framework presents several challenges. A primary issue lies in ensuring smooth and coherent connections among its layers. Specifically, the relationships between inputs and outputs must be well-defined and exhibit stable behavior to produce reliable results. If these connections are compromised, the integrity of the optimization process may be affected, leading to numerical instability. For instance, instability in one layer’s output can propagate through the framework, ultimately undermining the optimization’s effectiveness.


Furthermore, the framework must be sufficiently adaptable to accommodate various investment objectives, such as maximizing returns, minimizing risk, or striking a balance between the two. This flexibility is essential for practitioners seeking to tailor strategies to specific investment goals while navigating complex market dynamics.

To address these challenges, we propose an extension to the End-to-End framework that enhances numerical stability within the optimization layer. Our approach introduces a bounded softmax layer, which ensures that the inputs to the optimization process remain stable and well-behaved. By constraining the layer’s outputs, we mitigate risks associated with extreme values that could disrupt the optimization process. This enhancement not only improves the reliability of the optimization results but also strengthens the framework’s overall performance, allowing it to accommodate a broader range of assets and investment scenarios.

To validate our proposed extension, we conduct a case study applying the bounded softmax layer to a well-established portfolio optimization methodology: risk budgeting. Risk budgeting is an increasingly popular method that allocates risk across assets based on their contributions to overall portfolio risk. By integrating our bounded softmax layer into this methodology, we demonstrate its effectiveness in enhancing the stability and reliability of the optimization process.

Our case study results indicate that the bounded softmax layer provides performance results that are $70\%$ to $100\%$ more reliable (implying a lower dispersion) than without this layer. Using performance metrics such as the Sharpe Ratio and the cumulative returns in the out-of-sample periods we obtain dispersion among results below $1\%$ for different starting values of the neural network weight that are a fundamental part of the end-to-end framework; this is a significant gain over the results replicated from \cite{Uysal2023}. 

Notably, this extension enables the framework to accommodate a larger number of assets without encountering substantial difficulties in achieving stable weight allocations. This capability is particularly valuable in today’s financial markets, where investors often seek to diversify portfolios across a wide range of asset classes to mitigate risk. The bounded softmax layer provides a robust mechanism for managing the numerical complexities of larger portfolios, ensuring that the optimization process remains both effective and  computationally efficient.

The paper is organized as follows. Section \ref{sec:literature} provides an overview of End-to-End modelling for portfolio optimization. Section \ref{sec:RiskBudgetingPortfolios} presents the fundamentals of risk budgeting portfolios and the \textit{plug-in} approach to portfolio optimization.  
Section \ref{sec:e2eframeworks} explains End-to-End frameworks emphasizing the shortcomings in terms of numerical stability and provides a detailed explanation of the proposed bounded softmax layer. 
Section \ref{sec:results} illustrates the main results with a case study based on real market data and a detailed explanation of the computational setup. 
Finally, Section \ref{sec:conclude} presents the conclusions and discuss possible extensions.

\section{Literature review}
\label{sec:literature}

End-to-end modeling in portfolio optimization represents a significant shift from traditional methods by integrating prediction and decision-making processes into an interconnected layered framework. While traditional portfolio optimization typically involves independent steps and employs the so-called \textit{plug-in} approach to link estimation and optimization stages, the layered approach provides a more cohesive modeling structure. This approach aims to address the limitations of conventional two-step processes by directly optimizing portfolio weights based on raw data inputs and targeting a specific performance outcome, thereby enhancing efficiency and potentially improving financial performance. The integration is expected to mitigate errors associated with separate prediction and optimization stages, leading to more robust investment decisions \cite{Anis2025}.

Recent studies propose initial implementations of the end-to-end modeling framework applied to mean-variance optimization \cite{Butler2022} or portfolio allocations defined by risk contributions \cite{Uysal2023}. In the former, the authors design a neural network with a layer for parametric estimation in a factor model, connected to a differentiable quadratic programming layer that solves a regularized mean-variance problem. In \cite{Uysal2023}, the authors present two approaches for portfolio construction: model-free and model-based networks, with the latter incorporating a risk budgeting model as an implicit layer. Risk budgets are estimated directly from the data and integrated into the portfolio optimization process to obtain optimal asset weights targeting specific outcomes such as maximizing cumulative ex-post returns or the Sharpe Ratio.

Empirical applications vary from seven to fifty assets, demonstrating promising results in controlling portfolio volatility and outperforming approaches that do not employ the end-to-end methodology, such as nominal risk parity or mean-variance optimization without regularization. The authors also highlight the need for further research to validate the efficacy and stability of these solutions.
Recent contributions provide notable extensions to these initial approaches. \cite{Zhang2021} propose a score-based layer that circumvents the reliance on specific moments such as mean and variance. The scores are linked to portfolio weights, determining the optimal stock combination according to a pre-specified financial outcome, such as cumulative returns adjusted for transaction costs. This method accommodates a large asset space (over one hundred stocks) and incorporates constraints on weights and the number of stocks (cardinality constraints). \cite{Anis2025} evaluate the complexities of including cardinality constraints as a layer between estimation and optimization layers. Their focus is on the cardinality-constrained minimum-variance portfolio optimization problem, utilizing a linear factor model for returns. This approach enhances the realism of portfolio construction while introducing computational challenges, as the optimization process must navigate a larger solution space while adhering to constraints on stock selection and weight limits. Their study initially considers fifty assets with varying degrees of cardinality constraints, targeting portfolios containing between 10 and 20 assets. The primary challenge is ensuring optimization efficiency while achieving optimal asset selection, as these constraints complicate the feasible region of the problem. To address this, advanced optimization techniques such as heuristic algorithms or mixed-integer programming can streamline the selection process while maintaining adherence to cardinality limits. This framework generates in-sample and out-of-sample portfolios with lower volatility and higher returns than traditional decoupled approaches when applied to real-world financial data.

\cite{Zhong2024} propose a Direct Sorted Portfolio Optimization (DSPO) end-to-end framework. This framework constructs characteristic-sorted portfolios directly from raw stock data without extensive manual feature engineering. Portfolio sorting is carried out using a Monotonic Logistic Regression loss function to rank assets optimally in relation to ex-post realized returns. The DSPO model achieved a rank information coefficient of $10.12$ and an accumulated return of $121.94\%$ on NYSE stocks.
Overall, end-to-end portfolio optimization and its multiple extensions show promising results. However, further refinement is necessary to enhance result stability and ensure clear replicability, particularly as some extensions introduce greater complexity.

For portfolio optimization problems there are different interpretations of instability problems. The classical ill-posed problem is a recurrent one in the "predict them optimize" strategy where generally the optimal weighted are instable and non-unique, meaning that small variation in the inputs that are \textit{plugged-in} to the optimization problem can generate important source of instability in the end-results: portfolio weights and performance. There is a large literature that proposed different correction to address the problems, for example the use of regularization techniques specially in the context of large portfolios (50-100 assets) \cite{Brodie2009}. L1 type regularization imposes sparsity in the final results, that is the effective number of actives positions (portfolio weights different from zero) is significantly reduced. Some extensions \cite{Dai2018} add shinkage to L1 regularization in order to improve performance.

Another approach is motivated by accounting for model risk and the intent is to obtain robust results, for example by explicitly accounting for uncertainty in the data. Robust portfolio optimization (RO) is a strategy in asset allocation that focuses on optimizing the performance of a portfolio under the worst-case scenarios of uncertain inputs. By embedding uncertainties into a deterministic framework, robust portfolio optimization provides a stable solution that tends to perform better out-of-sample compared to classical mean-variance portfolios \cite{Fabozzi2010,Xidonas2020}. More recently \cite{Georgantas2024} compares the performance of different RO versions of classical models: mean-variance with box (MVBU) or ellipsoid (MVEU) uncertainty , worst-case VaR and CVaR, worst-case Omega Ratio or a combination of some of them (multi-objective optimization). The performance of the models is evaluated considering the number of active assets and risk return trade-off. In general, RO performed better than their non-RO counterparts although their asset compositions are relatively similar in terms of diversification and turnover. Specifically, the MVBU model, while underperforming compared to the MV model during stable periods, showed consistently better results during stock market crises. Conversely, the MVEU model outperformed the MV model in five out of six indicators during the full test period but did not perform as well during crisis conditions.


\section{Risk Budgeting Portfolios}
\label{sec:RiskBudgetingPortfolios}

Consider and investment universe with $n$ assets. Let $\mathbf{w}=\left(w_1,w_2,...,w_n\right)^\top$ be the vector of portfolio weights and $\Sigma \in\mathbb{R}^{n\times n}$ be the covariance matrix of asset returns and the portfolio risk $\sigma$ is measured by the volatility, i.e. 
\begin{equation}
    \sigma(\mathbf{w})=\sqrt{\mathbf{w}^\top \Sigma\mathbf{w}}.
    \label{eq:volatility}
\end{equation}
The \textit{marginal risk contribution} of asset $i$ to the overall portfolio risk is given by the partial derivative of $\sigma(\mathbf{w})$ with respect to $w_i$: $$\dfrac{\partial\sigma(\mathbf{w})}{\partial w_i}=\dfrac{\left(\Sigma\mathbf{w}\right)_i}{\sigma (\mathbf{w})},$$
so that the \textit{total risk contribution} of asset $i$ is
\begin{equation}
    RC_i(\mathbf{w})=\dfrac{\partial\sigma(\mathbf{w})}{\partial w_i}=\dfrac{w_i\left(\Sigma\mathbf{w}\right)_i}{\sigma (\mathbf{w})}.
    \label{eq:riskContribution}
\end{equation}
Assume a pre-specified \textit{risk budget vector} $\mathbf{b}=\left(b_1,b_2,...,b_n\right)^\top$ where each $b_i\geq 0$ and $\mathbf{1}^\top \mathbf{b}=1$. This vector represents the target proportions of total portfolio risk that each asset should contribute.

The \textit{risk budgeting portfolio optimization problem} is then to determine the portfolio weights $\mathbf{w}$ such that the risk contribution of each asset is proportional to its risk budget. Formally this require that
\begin{equation}
    \dfrac{\partial RC_i(\mathbf{w})}{\sigma (\mathbf{w})}=b_i,\quad \text{for } i=1,2,...,n.
\end{equation}
Substituting the equations (\ref{eq:volatility}) and (\ref{eq:riskContribution}), this condition becomes:
\begin{equation}
    \dfrac{w_i\left(\Sigma\mathbf{w}\right)_i}{\mathbf{w}^\top \Sigma\mathbf{w}} \Leftrightarrow w_i\left(\Sigma\mathbf{w}\right)_i = b_i\mathbf{w}^\top \Sigma\mathbf{w}, \quad \text{for } i=1,2,...,n.
\end{equation}
subject to the standard portfolio constraint $\mathbf{1}^\top \mathbf{w}=1$.

This formulation ensures that each asset's contribution to the overall portfolio risk aligns with its predetermined risk budget. When all $b_i=\nicefrac{1}{n}$, the problem reduces to the well-known \textit{risk parity} portfolio optimization problem.

In many practical portfolio optimization settings, a "predict then optimize" strategy—also known as the classical plug‐in approach—is commonly employed. This methodology divides the decision-making process into two distinct stages. First, historical or auxiliary data are used to generate estimates of key model parameters, such as expected returns and covariances (or risk budgets for the risk budgeting portfolio optimization problem). These estimates are then directly substituted into an optimization framework in the second stage, where the optimal portfolio weights are determined under given constraints and objectives. Although this decoupled approach simplifies the process by treating prediction and optimization separately, it inherently links the quality of the final portfolio to the accuracy of the parameter estimates. The formal definition provided below captures this two-step process, detailing how parameter estimation feeds into the optimization of portfolio weights.

Let $\mathbf{\theta}$ denote the (unknown) parameters of the return distribution of $n$ assets. For instance, in the risk budgeting framework, one have $\mathbf{\theta}=(\mathbf{b},\Sigma)$. Given historical or auxiliary data $\mathbb{D}$, the investor uses a statistical model to compute an estimate $\hat{\mathbf{\theta}}=\mathcal{P}(\mathbb{D})$, where $\mathcal{P}$ is a prediction procedure e.g., maximum likelihood estimation, that maps the data $\mathbb{D}$ to an estimate in the parameter space $\mathbf{\Theta}$.
Now, suppose the investor's objective is to maximize a utility function or minimize a risk measure that depends on the portfolio weights $\mathbf{w}$ and the parameters $\theta$. We denote this objective function by $U(\mathbf{w};\mathbf{\theta})$. For example, in the risk budgeting setup the objectives might be written as
\begin{equation}
    U(\mathbf{w};\mathbf{b},\Sigma)=\sum_{i}^n \left(w_i\left(\Sigma\mathbf{w}\right)_i - b_i\mathbf{w}^\top \Sigma\mathbf{w}\right)^2
\end{equation}
In the plug-in approach, the estimated parameters $\hat{\mathbf{\theta}}$ are substituted or \textit{plugged in} into the objective function, and the optimal portfolio weights are chosen as if the estimated parameters were the true ones. For example, the risk budgeting portfolio optimization problem in the plug-in approach is defined as
\begin{equation}
    \mathbf{w}^*= \underset{w\in\mathcal{W}}{\text{arg min}} \; U(\mathbf{w};\mathbf{b},\Sigma)
    \label{eq:riskBudgetingProblem}
\end{equation}
where $\mathcal{W}=\displaystyle\left\{(w_1, w_2, ..., w_n)\in\mathbb{R}^n\left\vert \sum_{i=1}^n w_i = 1\right.\right\}.$ 




\section{End-to-End Frameworks}
\label{sec:e2eframeworks}

The end-to-end portfolio provides a unified framework for estimation, optimization allocation, and risk management in a classical static problem with two periods: the decision period $t$ and the realization period $t+1$.

Let $(X_{t},X_{t+1} \in \mathcal{X})$ represent the input and output variables in a predictive problem. We assume that the variables follow an unknown stationary distribution $\mathcal{D}$ whose moments provide relevant information for the portfolio problem. \\

The general end-to-end portfolio problem will be described as follows:
\begin{subequations}
\label{eq:optim}
\begin{align}
    \underset{\theta}{\text{minimize}}
        & \quad R_{\theta}(Z):=E_{t}[f(X_{t},X_{t+1},z^{*}(X_{t};\theta)]   \label{eq:outputDecode}\\
    \text{subject to} 
        & \quad z^{*}(X_{t},\theta)= \underset{z \in \mathcal{Z}}{\text{argmin}}\quad E_{D}[c(H(X_{t}),z)]   \label{eq:optLayer}
\end{align}
\end{subequations}

\textcolor{blue}{continue...}

\begin{figure}[H]
    \begin{center}
        \begin{tikzpicture}
            \Vertex[x=1,y=3,size=1, shape=rectangle,label=X_{t}, Math]{A}
            \Vertex[x=3,y=3,size=1, shape=rectangle,label=H(X_{t}), Math]{B}
            \Edge[Direct](A)(B)
            \Vertex[x=6.5,y=3,size=3.2,shape=rectangle,label= \displaystyle\underset{z}{\text{min}}\, E_{D}\left(c(H(X_{t});z)\right),fontscale=1, Math]{C}
            \Edge[Direct,label=\sigma,fontcolor=red, Math](B)(C)
            \Vertex[x=10.2,y=3,size=1,shape=rectangle,label=R_{\theta}(Z), Math]{D}
            \Edge[Direct,label=\frac{z}{\mid \mid z \mid \mid_{1}},fontcolor=red, Math](C)(D)
        \end{tikzpicture} 
        \caption{Neural network architecture for generic portfolio end-to-end problem}
    \label{fig:genE2EArch}
    \end{center}
\end{figure}
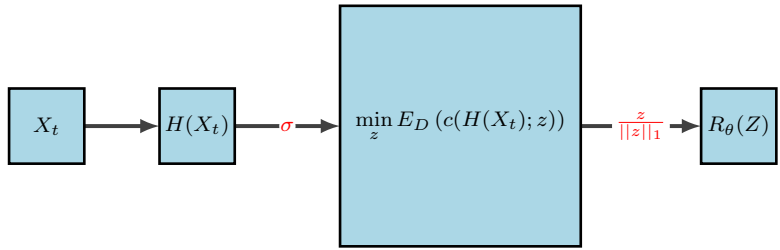


As explained in section \ref{sec:literature}, \cite{Uysal2023} proposed end-to-end modelling framework for risk budgeting portfolios. The empirical illustration is based on a portfolio of seven exchange-traded funds (ETFs). The framework employs a fully connected neural network that begins with a hidden layer of seven neurons activated by a Leaky ReLU function (with parameter $\alpha=0.1$). This is followed by a second hidden layer of seven neurons—representing the risk budget vector—utilizing a softmax activation function. Next, a CvxPyLayer is incorporated to solve the risk budgeting portfolio optimization problem, and a final layer computes a risk-reward function based on the normalized outcomes of this convex optimization.
The training algorithm adopts a rolling window approach. For the market data experiments, approximately 35\% of the dataset is used for training, 25\% for hyperparameter optimization, and the remaining 40\% for evaluation. The experiments report results for both the Sharpe ratio and cumulative return as risk-reward functions. In the following section we replicate \cite{Uysal2023} end-to-end framework and the main empirical results in order to illustrate the dependence of this particular setup on the initial conditions of the neural network weights used to implement the framework.    

\subsection{Stability to initial conditions}
\label{sec:stability}

To show the stability to initial conditions that exists in standard end-to-end frameworks, we replicate the model-based approach proposed in \cite{Uysal2023}, we ran the training algorithm with optimized parameters for each risk-reward function: Sharpe ratio and Cumulative return. Using $15$ different random seeds, i.e., $15$ different initializations of the neural network weights.

Figure \ref{fig:VarTrainUysal} illustrates the stability to initial conditions by showing the regions defined by the minimum ($\min(t)$) and maximum ($\max(t)$) cumulative returns achieved during the training period for each risk-reward function. Additionally, we plot the time series of midpoints of these regions, providing a better understanding for the measure of dispersion at each time step, that in our case is the range or dispersion estimator:  
$$v(t) = \max(t) - min(t).$$  
\begin{figure}[H]
    \centering
    \includegraphics[width=0.9\textwidth]{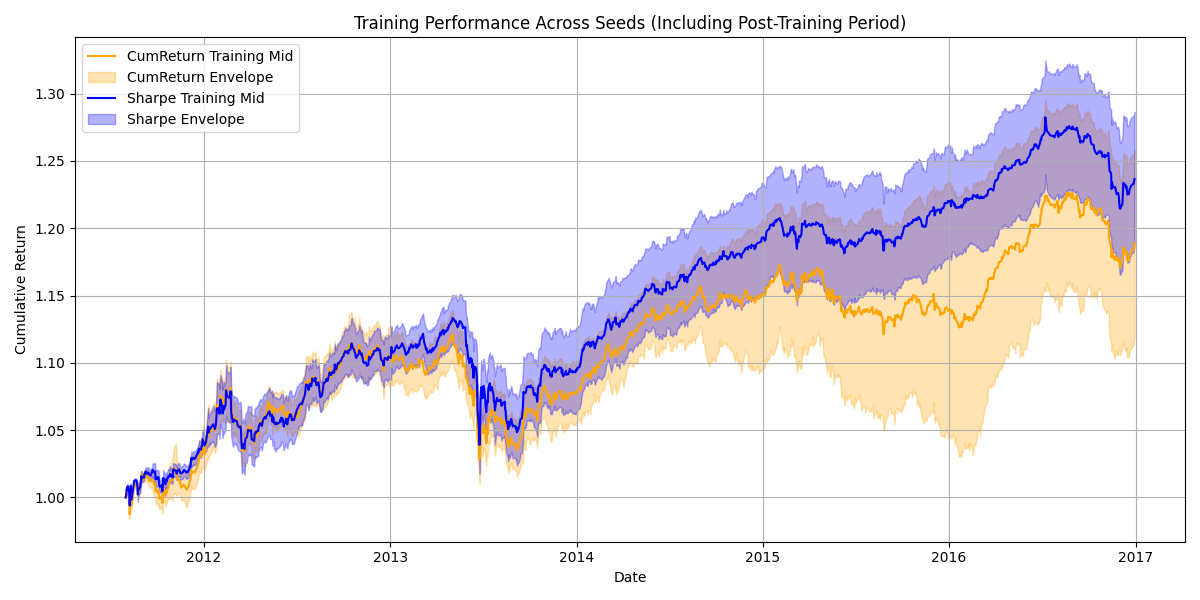}
    \caption{Dispersion of cumulative return during the training period for 15 seeds.}
    \label{fig:VarTrainUysal}
\end{figure}

It is important to clarify that we do not expect dispersion to be zero, as this could indicate an overfitting issue. Some variation across seeds is natural due to the stochastic optimization algorithm and the inherent bias-variance tradeoff. 
    
However, following the argument in \cite{Zhong2024}, we emphasize that in industry, portfolio allocation strategies should not exhibit excessive dispersion in results when using numerical optimization. This is why we aim to eliminate the dependence of the trained model’s performance on its weight initialization, which is a reasonable and practical objective.  

\begin{figure}[H]
    \centering
    \includegraphics[width=0.9\textwidth]{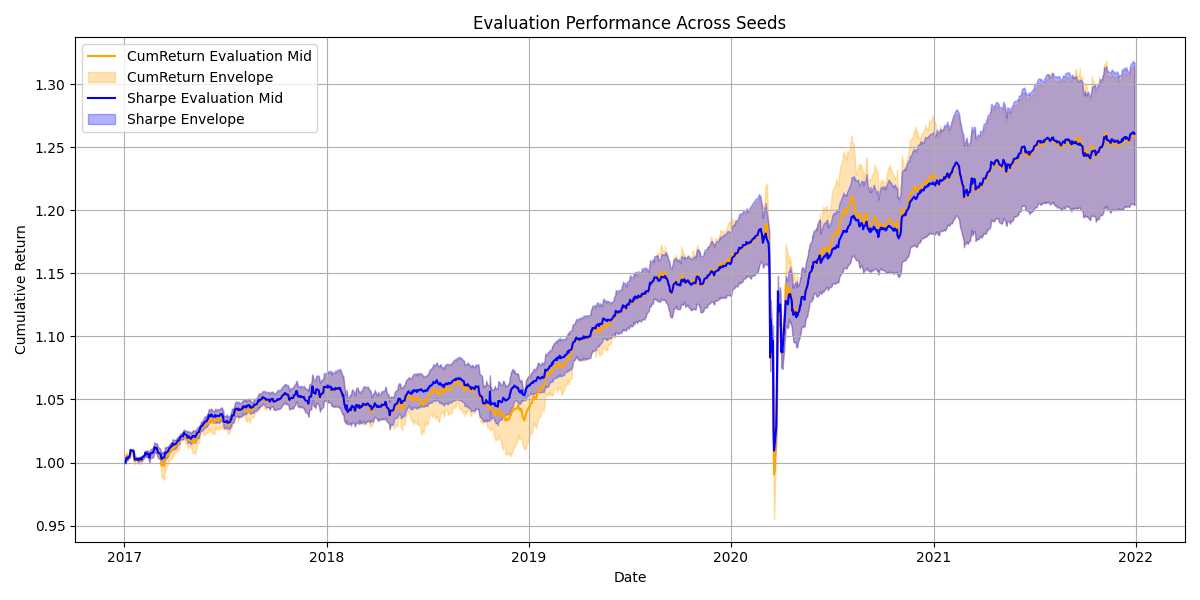}
    \caption{Dispersion of cumulative return during the evaluation period for 7 seeds.}
    \label{fig:VarTestUysal}
\end{figure}

Figure \ref{fig:VarTestUysal} presents the same setup for the evaluation period, further justifying our objective of improving the robustness of the risk-budgeting portfolio construction. Specifically, we aim for a significant reduction in the range $v(t)$ across all $t$ in the training period to ensure that decisions made during evaluation remain reliable. A summary of the dispersion results for this replication is shown in Table \ref{tab:sumUysal}.

\begin{table}[H]
  \centering
  \caption{Summary of dispersion results}\label{tab:sumUysal}
  \setlength\tabcolsep{0pt}
  \begin{tabular*}{\linewidth}{@{\extracolsep{\fill}}
                        l
                        S[table-format=-1.3]
                   *{2}{S[table-format= 1.3]}
                        S[table-format=-1.2]
                        S[table-format=-1.3]
                   *{2}{S[table-format= 1.3]}
                            }
    \specialrule{1.5pt}{2pt}{2pt}
    &   \multicolumn{2}{c}{Cumulative Return}   &   \multicolumn{2}{c}{Sharpe Ratio}   \\
    \cmidrule(lr){2-3}
    \cmidrule(lr){4-5}
    &   {\thead[b]{Training}}
        &   {\thead[b]{Evaluation}}
            &   {\thead[b]{Training}}
                            &   {\thead[b]{Evaluation}}                    \\
    \midrule
    $\text{Maximum }v(t)$  & 20.07\% & 11.52\% & 10.19\% & 11.32\%  \\
    $\text{Average }v(t)$  & 7.78\%  & 4.87\%  & 5.83\%  & 4.49\%  \\
    $\text{Last day }v(t)$ & 13.75\% & 10.87\% & 9.96\%  & 11.27\%  \\
    \specialrule{1.5pt}{2pt}{2pt}
  \end{tabular*}
\end{table}

\subsection{Stability-Focused End-to-End Framework}
\label{sec:e2eProposed}

Recent literature has been shifting towards developing robust end-to-end frameworks for portfolio allocation In \cite{Costa2023} the authors introduced a decision layer formulated to maximize an expression similar to a mean-variance utility function. However, instead of using portfolio variance as a risk measure, they define deviation risk measure estimated over a distribution of data that admits ambiguity concerning the true distribution of the data. Under this set up the decision layer can be framed as a minmax problem that can be efficiently managed through convex duality

One common critique of classical frameworks, such as Markowitz’s theory, is that they often yield low-variance but highly biased results. However, we cannot move to the opposite extreme, where high-variance but low-bias results dominate.  
    
Would an investor be willing to rely on a strategy that, due to high variability, might underperform the risk-parity benchmark? The dispersion of the Sharpe ratio in this setup suggests that such cases arise—not due to the fundamental risk-reward tradeoff of investing but rather as a consequence of the inherent limitations of early end-to-end frameworks.

To mitigate the high dispersion caused by instability to initial conditions, we propose several modifications to the previously studied framework. We optimize the number of neurons in the first hidden layer and we use a set of classical and aggressive learning rates to be tested. Next, we introduce a \textit{lower-bounded softmax function} to prevent risk budgets from vanishing before reaching the risk-budgeting optimization layer. This ensures that allocations remain meaningful and avoids degenerate solutions where certain assets receive near-zero weight allocations.

\subsubsection{Forward Pass}

Figure \ref{fig:architecture1} presents the computational graph of the proposed end-to-end framework. Given an initial dataset of asset features, the forward pass estimates the risk budget vector, utilizes this estimation to solve the risk budgeting optimization problem, normalizes the resulting allocation, and finally computes the loss function based on the obtained portfolio weights.

\begin{figure}[H]
    \centering
    \includegraphics[width=0.9\linewidth]{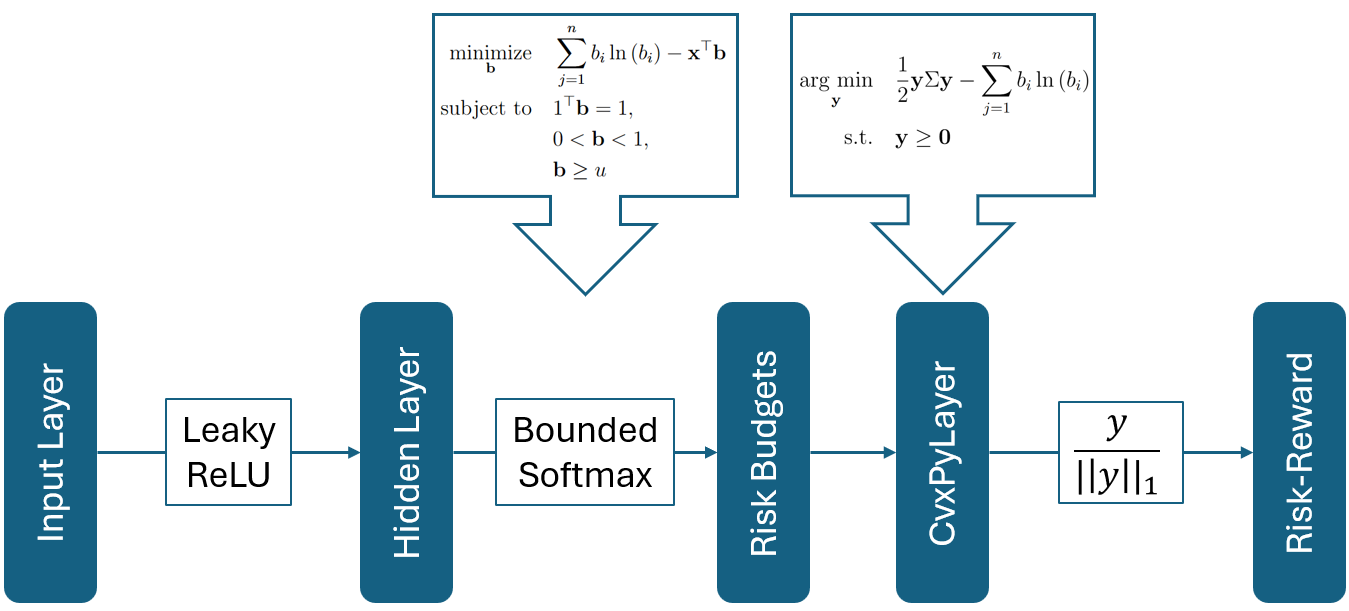}
    \caption{Scheme of the architecture.}
    \label{fig:architecture1}
\end{figure}

\textbf{Bounded Softmax}\\

The softmax activation function transforms an n-dimensional input vector ($\mathbf{x}$) into a probability-like distribution ($\mathbf{b}$), which, in this case, represents the investor's risk budgets. The standard softmax function is defined as:
\begin{equation}\label{eq:softmax}
    \mathbf{b}_i=\dfrac{e^{x_i}}{\displaystyle\sum_{j=1}^n e^{x_j}}
\end{equation}
The standard softmax allows risk budget values that are arbitrarily close to zero. However, as demonstrated in \cite{Roncalli2013}, the risk budgeting optimization problem has a unique solution only when all risk budgets ($\mathbf{b}$) are strictly positive. If any risk budget is exactly zero, the problem loses uniqueness, introducing instability in the optimization process.

We find that using an unbounded softmax activation to generate risk budgets leads to instability in both problem formulation and solutions. Specifically, certain training periods exhibited near-zero risk budgets (e.g., $10^{-6}$), which effectively behaved as computationally vanishing risk budgets, leading to unreliable allocations. This phenomenon was confirmed by solving the optimization problem outside the neural network framework for comparison.
From a financial perspective, enforcing a strictly positive lower bound on risk budgets is also meaningful. Portfolios with vanishing risk budgets imply complete exclusion of assets, contradicting diversification principles. Additionally, even minimal exposure to all assets ensures a better-conditioned optimization problem and reduces overfitting to specific asset groups.

To address this issue, we introduce a modified softmax activation function that incorporates a lower bound constraint. This is achieved by solving the following convex optimization problem:
\begin{equation}\label{eq:bounded}
    \begin{aligned}
        \underset{\mathbf{b}}{\text{minimize}} \quad & \sum_{j=1}^n b_i\ln{(b_i)}-\mathbf{x}^\top \mathbf{b} \\
        \text{subject to} \quad & 1^\top \mathbf{b}=1,\\
        & 0<\mathbf{b}<1, \\
        & \mathbf{b}\geq u
    \end{aligned}
\end{equation}
The solution to this problem results in a softmax-like function with a lower bound $u$. This result is formalized in Theorem \ref{th:boundedSoftmax}, with the proof provided in Appendix \ref{sec:boundedSoftmax}.

\begin{theorem}[Bounded Softmax]
    \label{th:boundedSoftmax}
    Let be $u$ the desired lower bound to the softmax function and $k=\#(A)$, with
    $$A=\left\{i\in\{1,2,...,n\}: \dfrac{e^{x_i}}{\sum_{j=1}^n e^{x_j}}\geq u\right\}$$
    The function 
    \begin{equation}\label{eq:bsoftmax}
        \mathbf{b}_i(\mathbf{x})=\left\{\begin{array}{cl}
            \dfrac{e^{x_i}}{\sum_{j\in A} e^{x_j}}\cdot\left(1-(n-k)u\right) & \text{if } i\in A\\
            u & \text{if } i\notin A
        \end{array}\right.
    \end{equation}
    is the unique solution to problem \ref{eq:bounded}.
\end{theorem}

\textbf{Risk Budgeting Optimization Layer}\\

\cite{Agrawal2019} introduced the concept of implicit layers in neural networks that solve convex optimization problems formulated as \textit{Disciplined Convex Programs}. The bounded softmax function presented above is an example of such an implicit layer. Similarly, we aim to integrate the risk budgeting optimization problem within the neural network framework by formulating it as a convex optimization problem.

\cite{Feng2015} showed that classical risk budgeting problem is non-convex, making direct optimization challenging, especially for large asset universes. However, \cite{Richard2019} demonstrated an equivalent convex formulation of the problem, allowing for computationally efficient optimization:

\begin{equation}\label{eq:convexForm}
    \begin{aligned}
        \mathbf{y}^* = \quad \underset{\mathbf{y}}{\text{arg min}} \quad & \dfrac{1}{2}\mathbf{y}\Sigma\mathbf{y}-\sum_{j=1}^n b_i\ln{(y_i)} \\
        \text{s.t.} \quad & \mathbf{y}\geq \mathbf{0}
    \end{aligned}
\end{equation}
The optimal portfolio weights are then obtained as: $$\mathbf{w}^*=\dfrac{y_i^*}{\displaystyle\sum_{j=1}^n y_j^*}$$
Notably, we use the quadratic form $\frac{1}{2}\mathbf{y}^\top\Sigma\mathbf{y}$ instead of the original square root formulation $\sqrt{\mathbf{y}^\top\Sigma\mathbf{y}}$ for computational efficiency. Since both functions are monotonic, they yield the same optimal solutions, while the quadratic form is more amenable to numerical optimization.

By embedding this convex optimization layer into the neural network, we enable efficient learning of risk budget allocations while ensuring that the obtained solutions remain well-conditioned and financially meaningful. The next section describes the backward pass, detailing how gradients propagate through the framework to facilitate end-to-end learning.

\subsubsection{Backward Pass}

The backward pass ensures that the gradients of all parameters within the end-to-end framework are correctly computed and propagated through the layers. This enables efficient optimization of the model by leveraging automatic differentiation and convex optimization layers.

Once the forward pass completes, the gradients of the loss function with respect to the model parameters must be computed. Let $\mathcal{R}$ denote the loss function. Using the chain rule of differentiation, the gradients propagate backward from the loss function to the earlier layers of the network. Mathematically, for any parameter $\theta$ in the network, we compute:
\begin{equation*}
    \dfrac{\partial\mathcal{R}}{\partial\theta} = \dfrac{\partial\mathcal{R}}{\partial\mathbf{w}}\cdot\dfrac{\partial\mathbf{w}}{\partial\mathbf{y}}\cdot\dfrac{\partial\mathbf{y}}{\partial\mathbf{b}}\cdot\dfrac{\partial\mathbf{b}}{\partial\mathbf{x}}\cdot\dfrac{\partial\mathbf{x}}{\partial\theta}
\end{equation*}

Each term in this expression represents a gradient computed at a specific stage of the computational pipeline.\\

\textbf{Differentiability of the Risk Budgeting Optimization Layer}\\

Since the risk budgeting allocation is obtained through a convex optimization layer, special care must be taken to ensure differentiability. The implicit function theorem \cite{Agrawal2019} provides a way to differentiate the optimal solution $\mathbf{y}^*$ with respect to the risk budget vector $\mathbf{b}$. Given the convex program:
$$\mathbf{y}^* = \quad \underset{\mathbf{y}}{\text{arg min}} \quad \dfrac{1}{2}\mathbf{y}\Sigma\mathbf{y}-\sum_{j=1}^n b_i\ln{(b_i)}$$
the gradient $\frac{\partial \mathbf{y}^*}{\partial\mathbf{b}}$ can be obtained by differentiating the KKT conditions. The resulting jacobian provides the necessary gradients for backpropagation.\\

\textbf{Differentiability of the Risk Budgeting Optimization Layer}\\

The bounded softmax layer introduces additional constraints by ensuring a lower bound $u$ on the risk budgets. Its gradient propagation follows from differentiating the optimality conditions of the constrained optimization problem \eqref{eq:bounded}. The Jacobian of the bounded softmax transformation is computed as:
\begin{equation*}
    \dfrac{\partial\mathbf{b}}{\partial\mathbf{x}} = J(\mathbf{x})
\end{equation*}
where $J(\mathbf{x})$ is derived in Appendix \ref{sec:boundedSoftmax}. This Jacobian is then used to propagate gradients to earlier layers.\\

\textbf{Backpropagation}\\

Using the computed Jacobians, the gradients of the loss function are propagated through the model. The backward pass consists of:
\begin{enumerate}
    \item $\dfrac{\partial\mathcal{R}}{\partial\mathbf{w}}$: Loss function gradient with respect to the normalized allocation. 
    \item $\dfrac{\partial\mathbf{w}}{\partial\mathbf{y}}$: Allocation gradient with respect to the risk budgeting layer solution.
    \item $\dfrac{\partial\mathbf{y}}{\partial\mathbf{b}}$: Propagating gradients through the risk budgeting optimization layer. 
    \item $\dfrac{\partial\mathbf{b}}{\partial\mathbf{x}}$: Propagating gradients through the bounded softmax layer. 
    \item $\dfrac{\partial\mathbf{x}}{\partial\theta}$: Propagating gradients to earlier neural network layers for parameter updates. 
\end{enumerate}

\subsubsection{Loss Function}

There are numerous ways to define a loss function for an end-to-end framework. For example, \cite{Uysal2023} and \cite{Anis2025} employed decision-based measures, specifically the Sharpe Ratio, to guide the learning process. In their approach, the model is trained by maximizing the Sharpe Ratio, effectively making the loss function its negative counterpart:
\begin{equation*}
    \mathcal{R_{SR}(\mathbf{w})}=-\dfrac{\mathbb{E}[r^w]}{\sigma^w}
\end{equation*}
where $r_w$ represents the portfolio return for allocation $\mathbf{w}$ and $\sigma_w$ its standard deviation. We found this loss function gives a dual formulation aligned with the nature of risk budgeting, where the convex optimization layer minimizes portfolio variance, while the neural network seeks to maximize expected returns. By optimizing these complementary objectives, the model achieves a balance between risk and return, leading to robust portfolio allocations.

For the analysis of stability we will also use the Sharpe Ratio and the cumulative return to define the loss function to train two different architectures, having a duality in the end-to-end problem: the risk budgeting minimize the variance and the neural network maximize the return, using the cumulative return as our standard performance metric for results.
\begin{equation*}
    \mathcal{R_{CR}(\mathbf{w})}=-\prod_{t=1}^T \left(1+r^w_t\right)
\end{equation*}

\section{Computational Experiments}
\label{sec:results}

\subsection{Computational Setup}

All computations were conducted using the \texttt{Python} programming language. The bounded softmax and risk budgeting optimization problems were implemented as implicit layers within a neural network, following the architecture presented in Figure \ref{fig:architecture1}. These layers were formulated using disciplined convex programming principles and solved using the \texttt{ECOS 2.0.14} solver, ensuring numerical stability and efficient optimization.

The convex optimization layers were implemented using the \texttt{cvxpylayers} package within the \texttt{PyTorch} framework, allowing for seamless integration of differentiable convex programs into the end-to-end learning process. This approach enables efficient gradient propagation through the optimization layers, facilitating stable training and robust convergence.

The neural network training process was executed on a machine equipped with a 13th Gen Intel(R) Core(TM) i7-13620H 2.40GHz processor and 16.0 GB RAM.

\subsection{Addressing stability concerns}

This first set of results demonstrates the improvements in stability achieved by our proposed framework. To assess this improvement, we use the framework introduced in \cite{Uysal2023} as a benchmark for comparison.

\subsubsection{Data}

For this stability analysis, we use daily returns from seven exchange-traded funds (ETFs) to represent broad stock, bond, and commodity market conditions. The selected ETFs include: VTI (Vanguard Total Stock Market ETF), IWM (iShares Russell 2000 ETF), AGG (iShares Core U.S. Aggregate Bond ETF), LQD (iShares iBoxx Investment Grade Corporate Bond ETF), MUB (iShares National Muni Bond ETF), DBC (Invesco DB Commodity Index Tracking Fund), and GLD (SPDR Gold Shares). The dataset spans from 2011 to 2021, with daily frequency, and was sourced from Yahoo Finance.

To construct the dataset, we compute daily log-returns based on adjusted closing prices to account for dividends and stock splits. Any missing values due to market closures were forward-filled to maintain data continuity. The dataset is split into two periods:
\begin{itemize}
    \item In-Sample Period (2011-2016/12): Used for training and hyperparameter tuning.
    \item Out-of-Sample Period (2017-2021/12): Used for evaluating model performance.
\end{itemize}

For hyperparameter selection, we further subdivide the in-sample data, using 2011-2014/12 for training and 2015-2016/12 for validation. The ETFs were selected to ensure diverse market representation, allowing the framework to adapt to various economic conditions. We assume a frictionless market and a static rebalancing strategy.

\subsubsection{In-Sample Results}

To determine the optimal hyperparameters, we perform a random grid search with 100 iterations during the validation period (2015-01-01 to 2016-12-31). The hyperparameter search space is defined as follows:
\begin{itemize}
    \item Hidden layer neurons: \{7, 16, 32\},
    \item Learning rate: \{0.05, 0.1, 0.5, 1.0, 5.0, 10, 50, 100, 150, 200, 250, 300, 350, 400\},
    \item Training steps: \{5, 10, 15, 20, 25, 30\}
\end{itemize}

The learning rate is a hyperparameter that determines how much a model's parameters (weights and biases) are adjusted in response to the loss at each training step. Specifically, during backpropagation, the network computes the gradient of the loss function with respect to each parameter, and the learning rate scales these gradients before updating the parameters. A high learning rate can speed up the initial learning process by taking larger steps in the parameter space; however, it may also cause instability and overshooting of the minimum, leading to divergent or oscillatory training. Conversely, a low learning rate results in smaller, more controlled updates that often promote more reliable convergence—but it can slow down training significantly and may cause the model to become trapped in suboptimal local minima.

Optimizing the number of neurons in a neural network is equally critical because it directly affects the model's capacity to learn and generalize. Too few neurons may fail to capture the complexity of the data, resulting in underfitting where the model misses important patterns. On the other hand, too many neurons can lead to overfitting, where the model memorizes the training data (including noise) and performs poorly on unseen data. The objective is to strike a balance that provides enough complexity to model the data effectively without overcomplicating the network, ensuring robust performance on both training and validation datasets.

Finally, training steps refer to the number of iterations during which the model updates its parameters. At each step, a batch of data is processed, the loss is computed, and gradients are used to adjust the model's weights. The number of training steps determines how long the model has to learn: too few steps may result in underfitting, while too many can cause overfitting. To maintain a simple architecture, we decided to optimize this parameter as a fixed hyperparameter. However, another alternative is to run the training algorithm for the number of steps required to reach a desired loss level.

We include “aggressive” learning rates (i.e., values $\geq 1.0$) to compare with those used by \cite{Uysal2023}. Figure~\ref{fig:OptSR} illustrates the results of our hyperparameter optimization. Each line in the figure corresponds to one model configuration, showing how the chosen hyperparameter values—such as learning rate, number of neurons in the hidden layer, and number of training steps—lead to a particular sharpe ratio at the end of the validation period. The red line represents the best model, which we identify by its highest Sharpe ratio during validation. Meanwhile, the four blue lines indicate other model configurations that also performed well, completing the top five models (i.e. highest sharpe ratio) evaluated during the hyperparameter optimization.


\begin{figure}[H]
    \centering
    \includegraphics[width=0.9\textwidth]{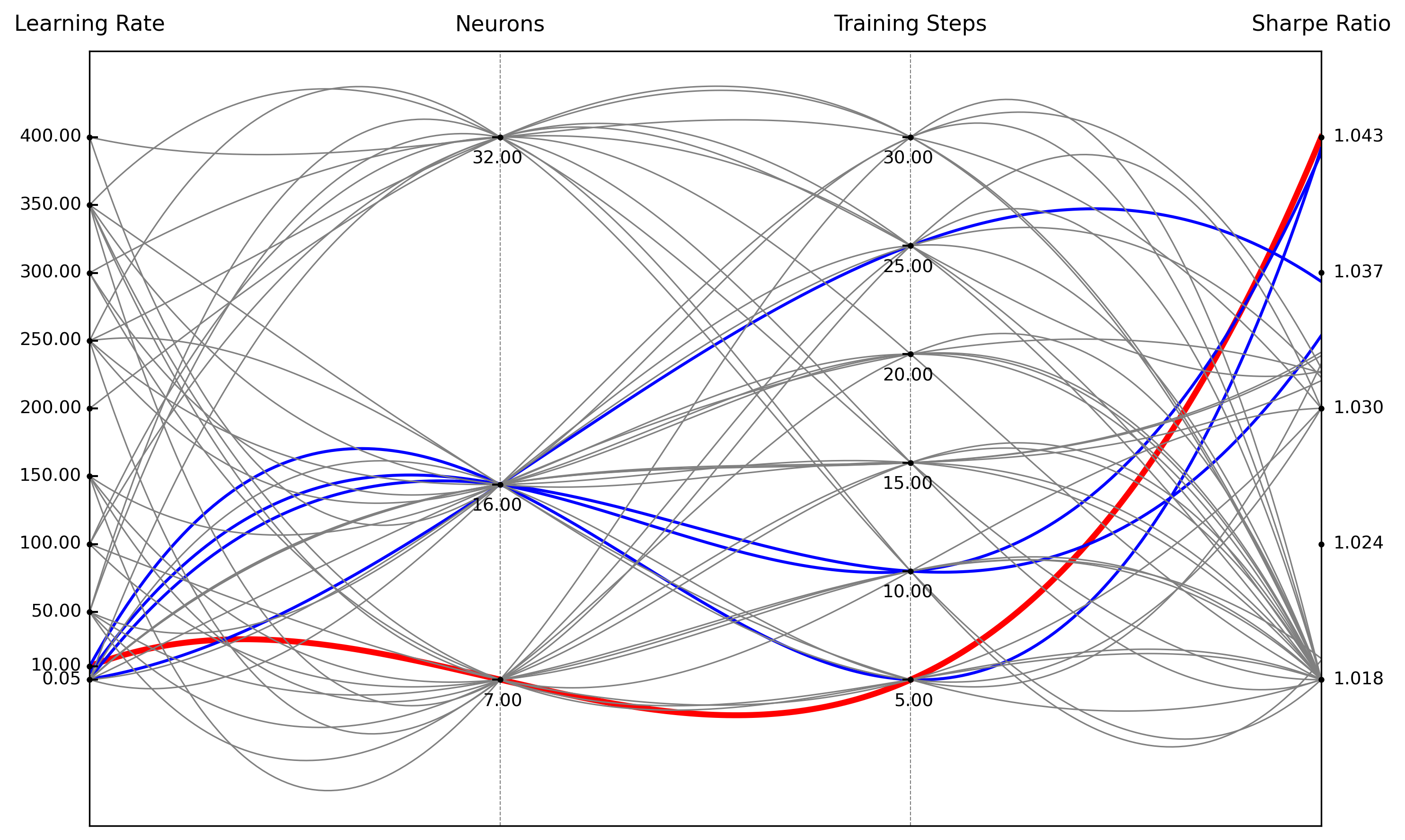}
    \caption{Hyperparameter Tuning results for Sharpe Ratio}
    \label{fig:OptSR}
\end{figure}

The best model (red line) is an end-to-end framework featuring a hidden layer with $7$ neurons, a learning rate of $10$, and trained over $5$ steps. In contrast, the other top-performing models (blue lines) utilize a hidden layer with 16 neurons and require a greater number of training steps, resulting in a more computationally complex architecture. A similar process is carried out using cumulative return as the loss function, yielding the results shown in Figure \ref{fig:OptCR}. In this scenario, the best model employs $16$ neurons and requires more training steps. In contrast, the model trained with the Sharpe Ratio converges faster, which is expected given its lower neuron count and reduced computational complexity.

To assess the stability of the optimized models, we repeat the training process across 15 different random seeds and measure the dispersion of cumulative returns during the training period. The results for this stability analysis are presented in Figure \ref{fig:VarTrainProposed}.

\begin{figure}[H]
    \centering
    \includegraphics[width=0.9\textwidth]{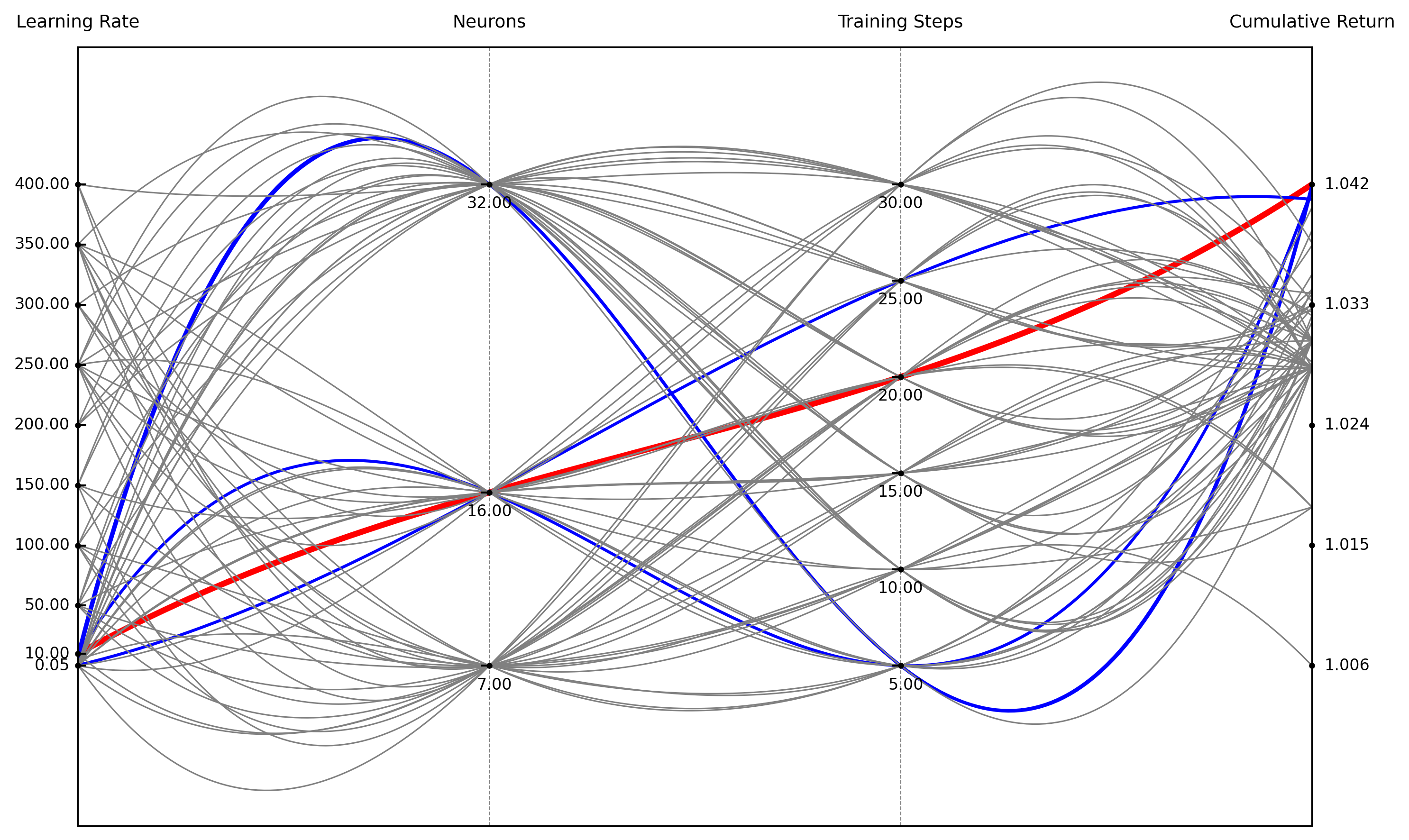}
    \caption{Hyperparameter Tuning results for Cumulative Return}
    \label{fig:OptCR}
\end{figure}

\begin{figure}[H]
    \centering
    \includegraphics[width=0.9\textwidth]{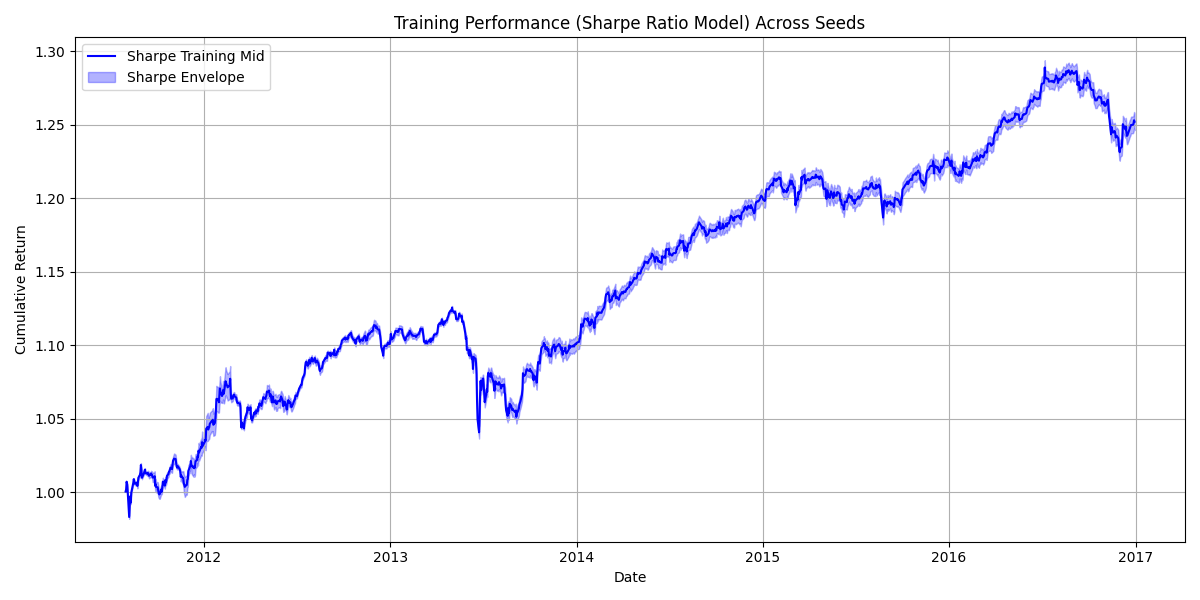}
    \caption{Dispersion of cumulative return during training period for sharpe-ratio.}
    \label{fig:VarTrainProposed}
\end{figure}
    
A comparison of the results of the dispersion (measured as the range statistic) between our framework and the replication of the existing approach in \cite{Uysal2023} is summarized in Table \ref{tab:sumInSample}. The results demonstrate a significant reduction in dispersion during both the training period and the evaluation period, highlighting the stability improvements introduced by our framework when using the Sharpe Ratio as the risk-reward function.

We repeat this process for the model trained using the cumulative return, we show the graph that represent the results for this case in Figure \ref{fig:VarTrainProposedCRet}. Although for this risk-reward function the Uysal model was more stable (lower dispersion than sharpe ratio's model) we still are getting better results.

\begin{figure}[H]
    \centering
    \includegraphics[width=0.9\textwidth]{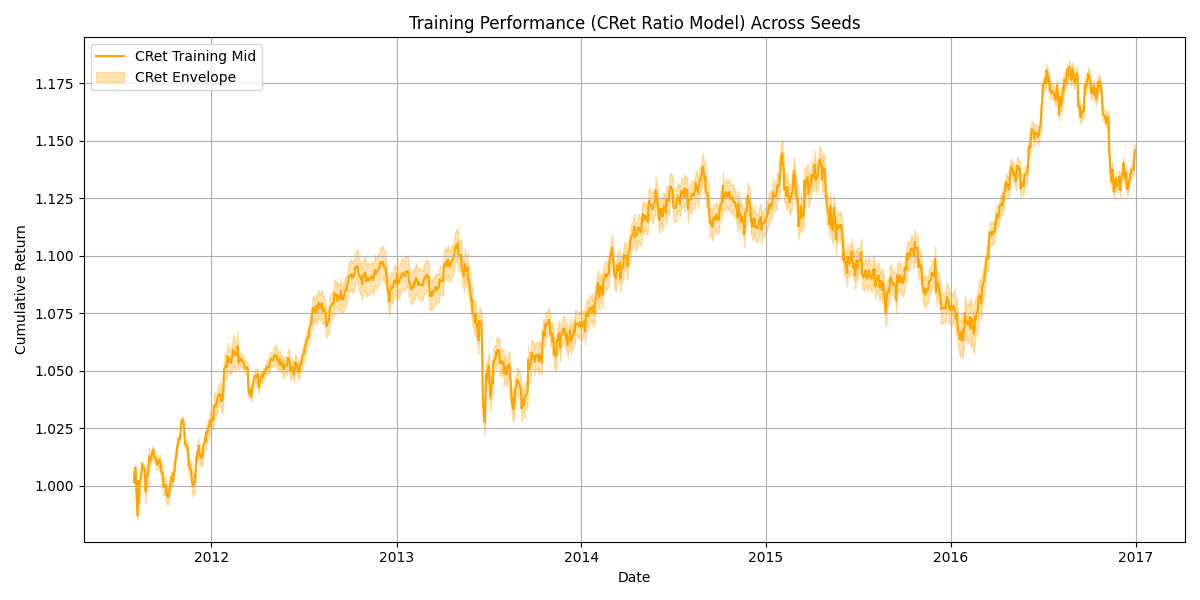}
    \caption{Dispersion of cumulative return during training period for Cumulative Return.}
    \label{fig:VarTrainProposedCRet}
\end{figure}

\begin{table}[H]
  \centering
  \caption{Summary of dispersion statistics for In-Sample period.}
  \label{tab:sumInSample}
  \setlength\tabcolsep{0pt}
  \begin{tabular*}{\linewidth}{@{\extracolsep{\fill}}
                        l
                        S[table-format=-1.3]
                   *{2}{S[table-format= 1.3]}
                        S[table-format=-1.2]
                        S[table-format=-1.3]
                   *{2}{S[table-format= 1.3]}
                            }
    \specialrule{1.5pt}{2pt}{2pt}
    &   \multicolumn{2}{c}{Cumulative Return}   &   \multicolumn{2}{c}{Sharpe Ratio}   \\
    \cmidrule(lr){2-3}
    \cmidrule(lr){4-5}
    &   {\thead[b]{Uysal}}
        &   {\thead[b]{Proposed}}
            &   {\thead[b]{Uysal}}
                            &   {\thead[b]{Proposed}}                    \\
    \midrule
    $\text{Maximum }v(t)$  & 20.07\% & \textbf{1.71}\% & 10.19\% & \textbf{1.91}\%  \\
    $\text{Average }v(t)$  & 7.78\%  & \textbf{1.02}\%  & 5.83\%  & \textbf{0.81}\%  \\
    $\text{Last day }v(t)$ & 13.75\% & \textbf{0.56}\% & 9.96\%  & \textbf{1.13}\%  \\
    \specialrule{1.5pt}{2pt}{2pt}
  \end{tabular*}
\end{table}

Finally, we assess whether the observed improvements in stability come at the expense of performance, measured in terms of cumulative return. Figure \ref{fig:CrSrRp} has the time series of average cumulative return for both best models across $15$ seeds and confirms that, despite the stabilization of results, we expect an end-to-end framework trained with Sharpe-Ratio that outperforms during the whole In-Sample period the risk-parity benchmark. Moreover, our architecture exhibits convergence independent to initial conditions, ensuring a global optimum solution.

By other hand, the end-to-end framework trained with cumulative returns seems to be more sensitive to market conditions and giving similar outcomes that the risk parity strategy. This result is not surprising since the training algorithm of neural network parameters are not taking into account the return-variance trade off.

\begin{figure}[H]
    \centering
    \includegraphics[width=0.9\textwidth]{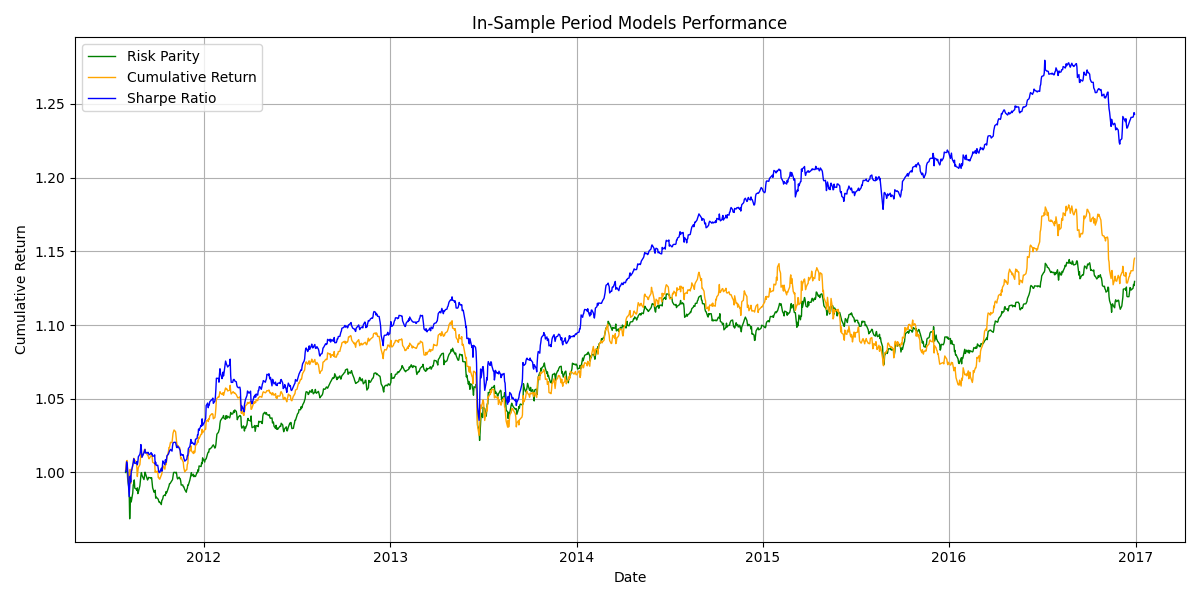}
    \caption{Mean Cumulative return of optimized end-to-end models.}
    \label{fig:CrSrRp}
\end{figure}

\subsubsection{Out-of-Sample Results}

To assess the robustness and generalization capability of our proposed framework, we evaluate its performance during the out-of-sample period (2017-2021/12). Specifically, we analyze the dispersion of cumulative returns across multiple random seeds to examine the stability of the model to different initializations.

Figure \ref{fig:VarTestProposed} presents the dispersion of cumulative return trajectories obtained from 15 different random seeds during the evaluation period. The results indicate that our framework maintains the performance in out-of-sample conditions, mitigating the instability to initial conditions observed in the benchmark model. Moreover, it shows that we do not have and overfitting problem using this data split.

\begin{figure}[H]
    \centering
    \includegraphics[width=0.9\textwidth]{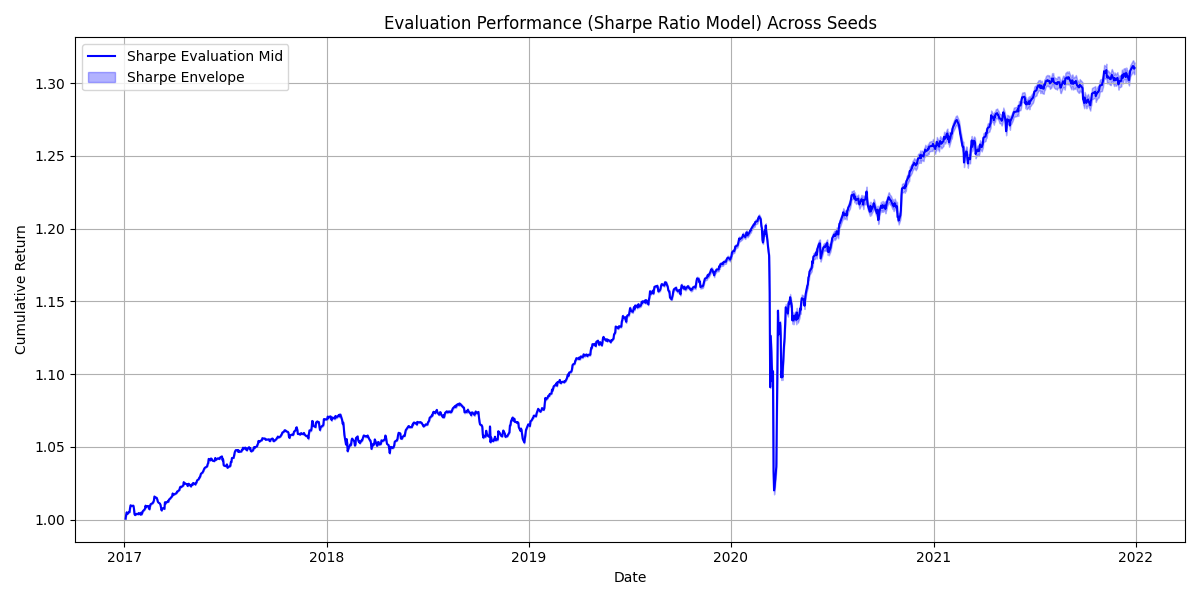}
    \caption{Dispersion of cumulative returns during evaluation period for sharpe-ratio.}
    \label{fig:VarTestProposed}
\end{figure}

To further quantify the stability improvements, Table \ref{tab:sumOutSample} compares the dispersion results of our framework against those of the benchmark model from \cite{Uysal2023}. The proposed model exhibits significantly lower dispersion throughout the evaluation period, reinforcing its capacity to produce consistent portfolio allocations.

Again, this process is also done for the model trained with cumulative return, obtaining the results in Figure \ref{fig:VarTestProposedCRet} and Table \ref{tab:sumOutSample}.

\begin{figure}[H]
    \centering
    \includegraphics[width=0.9\textwidth]{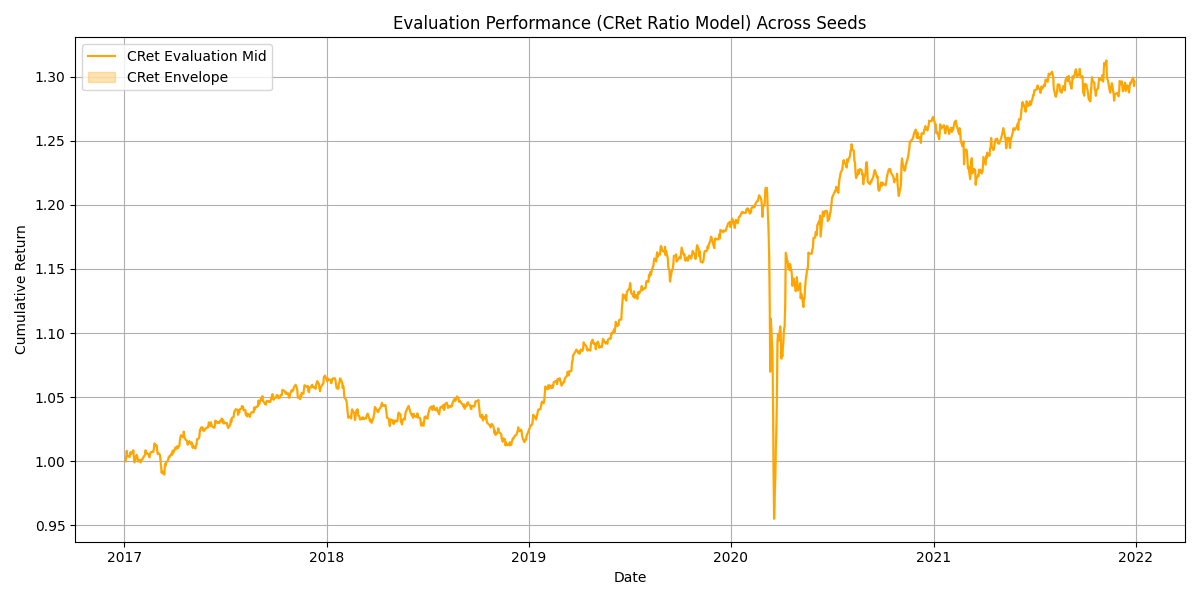}
    \caption{Dispersion of cumulative returns during evaluation period for sharpe-ratio.}
    \label{fig:VarTestProposedCRet}
\end{figure}
\begin{table}[H]
  \centering
  \caption{Summary of dispersion statistics for Out-of-Sample period.}\label{tab:sumOutSample}
  \setlength\tabcolsep{0pt}
  \begin{tabular*}{\linewidth}{@{\extracolsep{\fill}}
                        l
                        S[table-format=-1.3]
                   *{2}{S[table-format= 1.3]}
                        S[table-format=-1.2]
                        S[table-format=-1.3]
                   *{2}{S[table-format= 1.3]}
                            }
    \specialrule{1.5pt}{2pt}{2pt}
    &   \multicolumn{2}{c}{Cumulative Return}   &   \multicolumn{2}{c}{Sharpe Ratio}   \\
    \cmidrule(lr){2-3}
    \cmidrule(lr){4-5}
    &   {\thead[b]{Uysal}}
        &   {\thead[b]{Proposed}}
            &   {\thead[b]{Uysal}}
                            &   {\thead[b]{Proposed}}                    \\
    \midrule
    $\text{Maximum }v(t)$  & 11.52\% & \textbf{0.0183}\% & 11.32\% & \textbf{0.87}\%  \\
    $\text{Average }v(t)$  & 4.87\%  & \textbf{0.0073}\%  & 4.49\%  & \textbf{0.29}\%  \\
    $\text{Last day }v(t)$ & 10.87\% & \textbf{0.0181}\% & 11.27\%  & \textbf{0.78}\%  \\
    \specialrule{1.5pt}{2pt}{2pt}
  \end{tabular*}
\end{table}

These findings confirm that our end-to-end framework not only stabilizes results in the training phase but also preserves its advantages in an out-of-sample setting. The reduced dispersion in cumulative returns suggests that our approach enhances consistency and robustness, addressing the instability observed in the benchmark framework. Moreover, the sensitivity to market conditions of the model trained with cumulative returns is again an observable feature in Figure \ref{fig:CrSrRpTest}: After 2019, the model gains momentum and reaches the performance of the Sharpe Ratio–trained model starting in the third quarter of the year; from that point on, both models recover similarly from the shock caused by the COVID‑19 pandemic. Most important, both end-to-end frameworks still improve the performance of risk parity strategies, supporting the results in the existing literature.

\begin{figure}[H]
    \centering
    \includegraphics[width=0.9\textwidth]{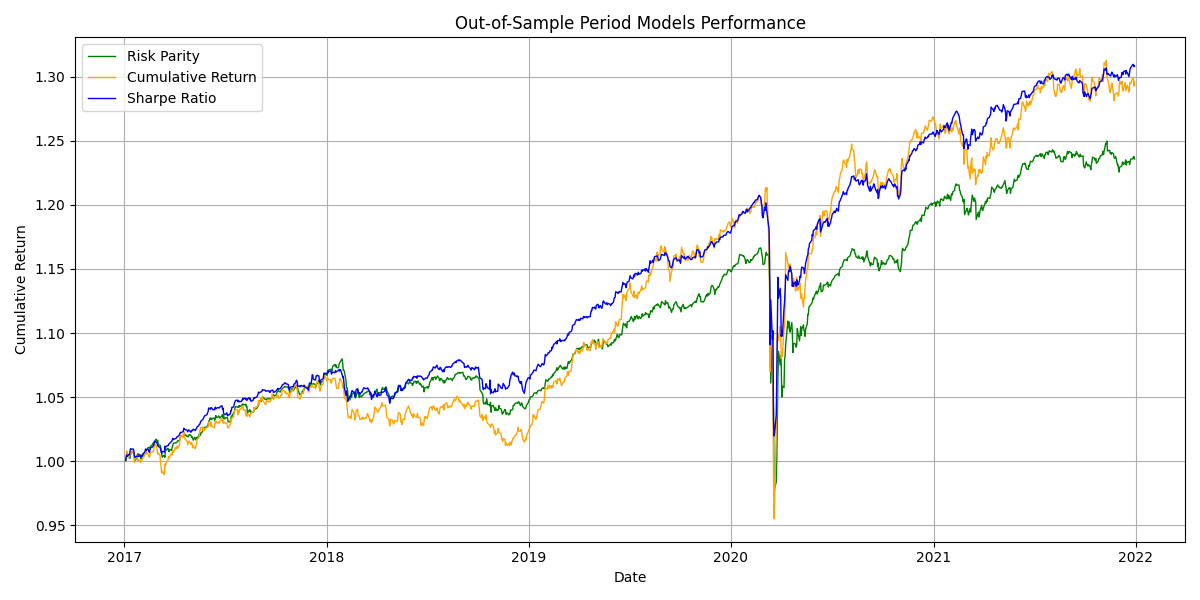}
    \caption{Mean Cumulative return of optimized end-to-end models.}
    \label{fig:CrSrRpTest}
\end{figure}

\section{Conclusion}
\label{sec:conclude}

In this study, we explored the existing numerical instability of end-to-end models with respect to the initial weight conditions in neural networks, aiming to assess the feasibility of constructing an environment capable of yielding reliable investment strategies. Our focus on risk-budgeting portfolios is particularly relevant given the current debate surrounding risk-parity strategies in financial markets.

Our findings reveal that the proposed model significantly reduces the dispersion of the ``optimal" solution when trained under different initial weight conditions. This reduction in dispersion indicates that the architecture converges toward a global solution, effectively avoiding the pitfalls of local optima. We attribute this behavior to our model’s treatment of risk factors: by constraining these factors from reaching extreme values, we maintain the uniqueness of the risk-budgeting solution.

These results underscore the importance of accounting for the specific dynamics of different portfolio optimization models when designing end-to-end architectures. They also pave the way for the broader implementation of deep learning techniques in the financial industry, as a globally convergent model increases confidence in the portfolio management process.

Despite these promising outcomes, our study identifies limitations in the current optimization solvers. Future work should focus on integrating more robust solvers, such as those available through CvxPyLayer, to enhance performance and extend the applicability of our approach. Additionally, exploring the scalability of this methodology to larger and more diverse asset classes, as well as its performance in live market environments, would be valuable. A comparative analysis with traditional portfolio optimization methods could further highlight the advantages of our deep learning-based strategy.

\subsection{Work In Progress}

We plan to extend the use of risk-reward functions to a broader range of satisfaction measures by following the framework presented in \cite{Meucci2007}. Specifically, we will define the loss function in an unsupervised setting and evaluate our architecture using representative measures such as the certainty equivalent (expected utility indices), Value at Risk (quantile-based indices), and Expected Shortfall (coherent indices). To assess robustness, we will test the model’s performance using real market conditions from recent years, and using several metrics in literature not only cumulative return e.g. maximum drawdown, Information ratio, etc.

By conducting these experiments, our goal is to develop a robust framework that remains stable across various market conditions and serves as an attractive portfolio management tool in today's data-rich environment. Additionally, we plan to explore the interpretability of the risk preferences learned by the model and extend the approach to other asset classes, further validating its adaptability and practical utility.

Another research avenue involves incorporating a module into the framework to estimate the covariance matrix using shrinkage techniques for factor models. This work is in its early stages, primarily due to the dimensionality challenges encountered when constructing portfolios with thousands of assets. By doing so, we aim to develop an optimal alternative to relying on the sample covariance matrix, which is notoriously unstable in high-dimensional settings. To our knowledge, the only attempt in the deep learning literature to improve covariance estimation for static portfolio settings is presented in \cite{Huynh2022}, which employs an asymmetric autoencoder model.

\appendix
\section{Bounded Softmax Function}
\label{sec:boundedSoftmax}

\begin{theorem}[Bounded Softmax]
Let be $u$ the desired lower bound to the softmax function and $k=\#(A)$, with
$$A=\left\{i\in\{1,2,...,n\}: \dfrac{e^{x_i}}{\sum_{j=1}^n e^{x_j}}\geq u\right\}$$
The function 
    \begin{equation}
        \mathbf{b}_i(\mathbf{x})=\left\{\begin{array}{cl}
            \dfrac{e^{x_i}}{\sum_{j\in A} e^{x_j}}\cdot\left(1-(k-n)u\right) & \text{if } i\in A\\
            u & \text{if } i\notin A
        \end{array}\right.
    \end{equation}
is the unique solution to problem \ref{eq:bounded}.
\end{theorem}
\begin{proof}
The KKT conditions for this problem are
\begin{itemize}
    \item{\textbf{Stationary:} 
    \begin{align*}
        \mathbf{0}= \dfrac{\partial}{\partial \mathbf{b}}\Biggl(\sum_{j=1}^n b_i\ln{(b_i)}-\mathbf{x}^\top \mathbf{b}-\sum_{j=1}^n \gamma_j b_j +&\sum_{j=1}^n \eta_j (b_j-1)  +\cdots\\
        &\cdots +\sum_{j=1}^n  \mu_j (u-b_j) + \lambda \left[1-\sum_{j=1}^n b_j\right]\Biggr)
    \end{align*}
    }
    \item{\textbf{Complementary slackness:}
        $$-b_j \gamma_j=0,\qquad \eta_j (b_j-1) = 0, \qquad \mu_j (u-b_j) = 0$$
    }
    \item{\textbf{Primal Feasibility:}
        $$1^\top \mathbf{b}=1,\qquad 0<\mathbf{b}<1,\qquad \mathbf{b}\geq u$$
    }
\end{itemize}
Then, from the stationary condition we get
\begin{align*}
    \dfrac{\partial}{\partial \mathbf{b}}\left( \mathbf{b}^\top\ln{\mathbf{(b)}}-\mathbf{x}^\top \mathbf{b}-\mathbf{\gamma}^\top \mathbf{b} +\eta^\top (\mathbf{b}-\mathbf{1}) +\mathbf{\mu}^\top (\mathbf{u}-\mathbf{b}) +\lambda \left[1-\mathbf{1}^\top \mathbf{b}\right] \right)&= \mathbf{0}\\
    \ln{\mathbf{(b)}}+\mathbf{1}-\mathbf{x}-\gamma +\eta -\mu - \lambda &= \mathbf{0}\\
    -\mathbf{1}+\mathbf{x}+\gamma -\eta +\mu + \lambda &=\ln{\mathbf{(b)}} 
\end{align*}
Therefore, for each $j=1,...,n$,
$$b_j = e^{-1+x_j +\gamma_j -\eta_j +\mu_j +\lambda}$$
Now, using the primal feasibility conditions, since $0<\mathbf{b}<1$, $\gamma_j=\eta_j=0$, thus for each $j=1,...,n$, $$b_j = e^{-1+x_j +\mu_j +\lambda}$$
Hence,
\begin{enumerate}
    \item{If $b_j > u$, $\mu_j=0$ and for each $j=1,...,n$, since $\mathbf{b}^\top \mathbf{1}=1$
    $$\sum_{j=1}^n b_j = \sum_{j=1}^n e^{-1+x_j  +\lambda}=1\quad\Longrightarrow \lambda = 1-\ln{\left(\sum_{j=1}^n e^{x_j}\right)}$$
    Then,
    \begin{align*}
        b_j &= e^{-1+x_j +\lambda}\\
        b_j &= e^{-1+x_j +\left(1-\ln{\left(\sum_{j=1}^n e^{x_j}\right)}\right)}\\
        b_j &= e^{x_j -\ln{\left(\sum_{j=1}^n e^{x_j}\right)}}\\
        b_j &= \dfrac{e^{x_j}}{\sum_{j=1}^n e^{x_j}}
    \end{align*}
    which is the classical softmax function (\ref{eq:softmax}). Moreover,
    $$A=\left\{i\in\{1,2,...,n\}: \dfrac{e^{x_i}}{\sum_{j=1}^n e^{x_j}}\geq u\right\}=\{1,2,...,n\}\Longrightarrow k=\#\left(A\right)=n$$
    and
    $$b_j =\dfrac{e^{x_j}}{\sum_{j=1}^n e^{x_j}}=\dfrac{e^{x_j}}{\sum_{j=1}^n e^{x_j}}\cdot\left(1-(k-n)u\right)$$
    }
    \item{If $b_j>u$ for $j=1,...,k$ and $b_j = u$ for $j=k+1, ..., n$, then
        $$\sum_{j=1}^k b_j = \sum_{j=1}^k e^{-1+x_j  +\lambda}=1-(n-k)u\Longrightarrow \lambda = 1+\ln{\left(1-(n-k)u\right)}-\ln{\left(\sum_{j=1}^k e^{x_j}\right)}$$
    Then, for $j=1,...,k$
    \begin{align*}
        b_j &= e^{-1+x_j +\lambda}\\
        b_j &= e^{-1+x_j +\left(1+\ln{\left(1-(n-k)u\right)}-\ln{\left(\sum_{j=1}^k e^{x_j}\right)}\right)}\\
        b_j &= e^{x_j +\ln{\left(1-(n-k)u\right)}-\ln{\left(\sum_{j=1}^k e^{x_j}\right)}}\\
        b_j &= \dfrac{e^{x_j}}{\sum_{j=1}^n e^{x_j}}\cdot \left(1-(n-k)u\right)
    \end{align*}
    which is the classical softmax function (\ref{eq:softmax}) applied over the $k$ coordinates left and multiplied by the remained weight. Thus,
    $$b_j=\left\{\begin{array}{cl}
        \dfrac{e^{x_j}}{\sum_{j=1}^n e^{x_j}}\cdot \left(1-(n-k)u\right) & \text{if } j\in A\\
        u & \text{if } j\notin A
    \end{array}\right.$$
    }
\end{enumerate}
\end{proof}

\end{document}